\documentclass[letterpaper, 10 pt, conference]{ieeeconf} 
\usepackage{xpatch, amsmath,amssymb,bm,mathtools}
\usepackage[capitalize]{cleveref}
\usepackage{xcolor}
\usepackage{fp, tikz, pgfplots}
\usepackage{siunitx}

\usetikzlibrary{arrows}
\usetikzlibrary{shapes}
\usetikzlibrary{patterns}
\usetikzlibrary{decorations.pathmorphing}
\usetikzlibrary{decorations.pathreplacing}
\usetikzlibrary{decorations.shapes}
\usetikzlibrary{decorations.text}
\usetikzlibrary{shapes.misc}
\usetikzlibrary{decorations.markings}
\usetikzlibrary{arrows.meta}
\usetikzlibrary{backgrounds}
\usepgfplotslibrary{external}
\usepgfplotslibrary{fillbetween}
\usepgfplotslibrary{statistics}
\usepackage[norelsize,ruled,noend]{algorithm2e}
\usepackage{setspace}

\IEEEoverridecommandlockouts                              

\newtheorem{assumption}{Assumption}
\newtheorem{lemma}{Lemma}
\newtheorem{remark}{Remark}
\newtheorem{theorem}{Theorem}
\newtheorem{definition}{Definition}
\newtheorem{corollary}{Corollary}

\newcommand{\approptoinn}[2]{\mathrel{\vcenter{
  \offinterlineskip\halign{\hfil$##$\cr
    #1\propto\cr\noalign{\kern2pt}#1\sim\cr\noalign{\kern-2pt}}}}}

\newcommand{\appropto}{\mathpalette\approptoinn\relax}

\makeatletter
\xpatchcmd{\proof}{\topsep0\p@\@plus0\p@\relax}{}{}{}
\makeatother

\pgfplotsset{width=10\columnwidth /10, compat = 1.13, 
	height = 55\columnwidth /100, grid= major, 
	legend cell align = left, ticklabel style = {font=\scriptsize},
	every axis label/.append style={font=\small},
	legend style = {font=\tiny},title style={yshift=-7pt, font = \small} }

\title{\LARGE \bf
Adaptive Low-Pass Filtering using Sliding Window Gaussian Processes
}

\author{Alejandro J. Ord\'o\~nez-Conejo$^{1*}$, Armin Lederer$^{2*}$ and Sandra Hirche$^{2}$
\thanks{$^{1}$Alejandro J. Ord\'o\~nez-Conejo is with Costa Rica Institute of Technology, 30101, Cartago, Costa Rica
        {\tt\small ajoseoc@gmail.com}}%
\thanks{$^{2}$Armin Lederer and Sandra Hirche are  with the Department of Electrical and Computer Engineering, Technical University of Munich, 80333 Munich, Germany
	{\tt\small [armin.lederer, 
	hirche]@tum.de}}%
\thanks{$^{*}$These authors contributed equally.}%
}

\begin{document}

\setlength{\textfloatsep}{14pt}
\setlength{\floatsep}{14pt}

\maketitle
\thispagestyle{empty}
\pagestyle{empty}

\begin{abstract}
When signals are measured through physical sensors, they are perturbed by noise. To reduce noise, low-pass filters are commonly employed in order to attenuate high frequency components in the incoming signal, regardless if they come from noise or the actual signal. Therefore, low-pass filters must be carefully tuned in order to avoid significant deterioration of the signal. This tuning requires prior knowledge about the signal, which is often not available in applications such as reinforcement learning or learning-based control. In order to overcome this limitation, we propose an adaptive low-pass filter based on Gaussian process regression. By considering a constant window of previous observations, updates and predictions fast enough for real-world filtering applications can be realized. Moreover, the online optimization of hyperparameters leads to an adaptation of the low-pass behavior, such that no prior tuning is necessary. We show that the estimation error of the proposed method is uniformly bounded, and demonstrate the flexibility and efficiency of the approach in several simulations.

\end{abstract}

\section{INTRODUCTION}

Control systems interact with the physical world by providing  signals to actuators based on measurements received from the plant. Since the physical quantities are measured using sensors, they are perturbed by noise resulting from various sources such as limited resolution of sensors, electromagnetic interference from the environment and quantization in digital to analog conversion. In order to reduce this measurement noise in signals, various forms of filters can be found in almost all real-world control systems.

These filters must inherently trade-off  noise attenuation and signal deterioration, since both goals cannot be achieved simultaneously \cite{Orfanidis2016}. As physical systems usually operate at comparatively low frequencies, this trade-off is usually achieved using low-pass filters, which strongly attenuate input signals above a cutoff frequency. They can be realized in digital systems based on linear system theory, which leads to approaches with infinite impulse response, e.g., Butterworth or Chebyshev filters, or with finite impulse response, e.g., moving average and Savitzky-Golay filters. As these filters attenuate input signals above the cutoff frequency regardless of their source, i.e., both the actual signal and the noise, they must be carefully tuned to avoid a deterioration of the signal. However, this requires prior knowledge about the signal and the noise distribution, which is often not available, in particular in learning systems. For example, in model-free reinforcement learning, no system model is known or even inferred, such that only little information about the system trajectories is known in advance.

When data can be processed off-line, such a lack of information has been overcome in many control applications by inferring models with Gaussian process regression \cite{Rasmussen2006}, which is a non-parametric supervised machine learning technique. Due to its Bayesian background, Gaussian processes naturally allow for an implicit bias-variance trade-off, and can efficiently learn from sparse data. Moreover, Gaussian processes also provide a measure for their predictive uncertainty, such that information about the reliability of Gaussian process models is available. This allows the derivation of prediction error bounds, which are based either on the theory of reproducing kernel Hilbert spaces \cite{Srinivas2012} or on Bayes theory~\cite{Lederer2019}. These advantageous properties have led to a wide spread usage of Gaussian processes in control, see, e.g., \cite{Kocijan2016}.

Despite their common application in control, Gaussian processes have been rarely considered in filtering problems so far. Probably most notable is their application as forward model in Kalman filters \cite{Deisenroth2009b, Ko2009}, but they have also been employed for attenuating noise in measured signals \cite{Calderhead2009, Wenk2020}. In the latter example, a Gaussian process is fitted to obtain a mapping between the measurement time and the signal state, which is exploited in a secondary step to identify the parameters of the signal generating differential equation using gradient matching. Thereby, the parameters of the differential equation can be obtained in a computational and data efficient way.

While such smoothing approaches are computationally efficient for parameter estimation, they are far too slow to be directly applicable in filtering problems with continuous streams of data due to the cubic complexity of Gaussian process regression. This limitation of Gaussian processes has been addressed by many authors recently. For example, the approaches in \cite{Gijsberts2013, Lederer2021b} aim to construct global approximations, which achieve update and prediction rates suitable for learning in control loops. However, this comes at the price of high memory requirements or significantly weaker theoretical guarantees. 
A different approach for enabling online learning with Gaussian processes relies on the restriction to a training set with constant number of samples, which gets updated after observing new samples \cite{Petelin2011}. By simply removing the oldest data in the set, the updates can be efficiently realized with constant complexity at high rates~\cite{Meier2016}. Moreover, these methods have been demonstrated to be able to continuously adapt to unknown system dynamics in experiments \cite{Meier2016a}.

Due to these promising results, we follow the idea of employing varying data sets with constant sizes for updating Gaussian processes online. Based on this idea, we develop an efficient algorithm called sliding window Gaussian process, which achieves high update and prediction rates. By employing this algorithm to learn a map between measurement times and signal values, it can be straightforwardly applied for filtering arbitrary signals. Using a novel approach, which does not rely on reproducing kernel Hilbert space theory and Bayesian properties, we prove that the filtered signal exhibits a bounded error under weak assumptions. The low-pass behavior of the proposed filter is investigated in simulations, in which we demonstrate its capability of adapting to input signals using the online optimization of hyperparameters of the Gaussian process. Finally, a comparison with commonly used filters clearly demonstrates an advantageous signal recovery over a wide range of frequencies.

The remainder of this work is structured as follows. In \cref{sec:prob}, the considered problem is formally described. The proposed low-pass filter using sliding window Gaussian processes is introduced and proven to provide bounded estimation errors in \cref{sec:filt}. In \cref{sec:eval}, the low-pass behavior of the sliding window Gaussian process is compared to commonly used filters. \cref{sec:conc} concludes this paper.\looseness=-1

\section{Problem Description}\label{sec:prob}
We consider a nonlinear autonomous system\footnote{Notation: 
Lower/upper case bold symbols denote vectors/matrices, 
$\mathbb{R}_+$/$\mathbb{N}_+$ all real/integer positive 
numbers, $\bm{I}_n$ the $n\times n$ identity matrix, $\bm{1}_n=\begin{bmatrix}1&0&\ldots&0\end{bmatrix}^T\in\mathbb{R}^n$ the unit vector, 
$\|\cdot\|$ the Euclidean norm, and  
$\lfloor\cdot\rfloor$ the floor operator. Scalar operators applied to vectors and matrices denote an element-wise operation.}
\begin{align}
    \dot{\bm{x}}&=\bm{f}(\bm{x},\bm{u},t),
\end{align}
where $\bm{x}\in\mathbb{X}\subseteq \mathbb{R}^d$ denotes the state, $\bm{u}\in\mathbb{U}\subseteq \mathbb{R}^p$ is the control input and $\bm{f}:\mathbb{X}\times\mathbb{U}\times\mathbb{R}_{0,+}\rightarrow\mathbb{R}^d$ is the nonlinear dynamic. In order to exclude systems, whose solutions have finite escape times, we make the following assumption.
\begin{assumption}
The function $\bm{f}(\cdot,\cdot,\cdot)$ is bounded for all $\bm{x}\in\mathbb{X}$, $\bm{u}\in\mathbb{U}$, $t\in\mathbb{R}_{0,+}$.
\end{assumption}
This assumption is satisfied, e.g., for time-invariant, continuous functions $f(\cdot,\cdot,\cdot)$ on compact domains $\mathbb{X}$ and $\mathbb{U}$, but also holds for many frequently found systems in practice such as, e.g., Euler-Lagrange systems with bounded control inputs. Therefore, it is not restrictive in general.

Since we cannot continuously measure the true state in real-world systems, we assume access to noisy measurements
\begin{align}
  \bm{y}^{(n)} &=\bm{x}(t^{(n)})+\bm{\epsilon}^{(n)}  
\end{align}
at times $t^{(n)}\in\mathbb{R}_{0,+}$, $n\in\mathbb{N}$, where $\bm{y}^{(n)}$ is the measured state at time $t^{(n)}$ disturbed by some noise $\bm{\epsilon}^{(n)}\in\mathbb{W}\subseteq\mathbb{R}^d$. For obtaining these samples, we consider a time-triggering scheme as formalized in the following.
\begin{assumption}\label{ass:data}
Data is sampled at regularly spaced time instances $t^{(n)}=n\tau$ with sampling time $\tau\in\mathbb{R}_{+}$ and aggregated into a data set $\mathbb{D}(t)=\{ (t^{(n)},\bm{y}^{(n)}) \}_{n=0}^{\lfloor\frac{t}{\tau}\rfloor}$.
\end{assumption}
Time-triggered sampling is probably the most common approach when designing control methods for continuous-time systems. Hence, this assumption is not restrictive in practice. 

Based on the obtained measurements aggregated in the data set $\mathbb{D}(t)$, we consider the problem of designing a filter for noise attenuation, which provides an estimate $\hat{\bm{x}}(t)$ of the noise free state $\bm{x}(t)$. While we cannot expect to achieve identity between $\hat{\bm{x}}(t)$ and $\bm{x}(t)$, we want to guarantee a small estimation error as defined in the following. 
\begin{definition}
The estimation error between the filtered state $\hat{\bm{x}}(t)$ and the noise free state $\bm{x}(t)$ is uniformly bounded by a function $c:\mathbb{R}_{0,+}\rightarrow\mathbb{R}_{0,+}$ if it satisfies
\begin{align}
    \|\hat{\bm{x}}(t)-\bm{x}(t)\|\leq c(t)\qquad \forall t\in\mathbb{R}_{0,+}.
\end{align}
\end{definition}
The goal is to derive a filter with guaranteed uniform error bound $c(\cdot)$, without posing additional assumptions on the dynamics $\bm{f}(\cdot,\cdot,\cdot)$ or the noise distribution. In order to achieve this goal, we derive a non-parametric online learning method based on Gaussian process regression and employ it for filtering the noisy data stream $\mathbb{D}(t)$ as explained in the following.\looseness=-1

\section{Adaptive Noise Attenuation with Sliding Window Gaussian Processes}\label{sec:filt}

Since the proposed adaptive filter, for noise attenuation with guaranteed estimation errors, is based on Gaussian process regression, we briefly explain this non-parametric supervised machine learning method in \cref{subsec:GP}. In order to make Gaussian processes applicable in filtering applications, in \cref{subsec:SW-GP}, we propose a fast approximation for online learning, which relies only on the most recent data. Finally, we apply this method as an adaptive low-pass filter and derive novel estimation error bounds in \cref{subsec:GP-filter}.

\subsection{Gaussian Process Regression}\label{subsec:GP}

Gaussian processes (GP) are an infinite collection of random variables that follow joint Gaussian distributions \cite{Rasmussen2006}. 
GP regression is used to model functions $ g:\mathbb{R}^d\rightarrow\mathbb{R} $ with inputs~$ \bm{z} $. A GP is defined by a prior mean function $ m:\mathbb{R}^d\rightarrow\mathbb{R} $ and a kernel function $ k:\mathbb{R}^d\times\mathbb{R}^d\rightarrow\mathbb{R}_{0,+} $. The prior mean $m(\cdot)$ can be used to include approximate, parametric models in the regression. Since we do not assume any knowledge of this kind, we set $m(\bm{z})=0$ for all $\bm{z}\in\mathbb{R}^d$ subsequently without loss of generality. The covariance function $k(\cdot,\cdot)$ is used to encode abstract prior knowledge such as smoothness and periodicity. When no such information is known in advance, the probably most frequently used covariance function is the squared exponential kernel
\begin{align}\label{eq:SE}
	k(\bm{x},\bm{x}')=\sigma_f^2\mathrm{exp}\left(-\sum\limits_{i=1}^{\rho}\frac{(x_i-x_i')^2}{2l_i^2}\right),
\end{align}
where $\sigma_f\in\mathbb{R}_{+}$ denotes the signal standard deviation and $\bm{l}=\begin{bmatrix}l_1&\cdots&l_d \end{bmatrix}^T$ are the length scales. These parameters form the hyperparameter vector $\bm{\theta}=\begin{bmatrix} \sigma_f& \bm{l}^T& \sigma_{\mathrm{on}} \end{bmatrix}$ together with the assumed standard deviation $\sigma_{\mathrm{on}}$ of the observation noise $\bm{\epsilon}$. Given a data set $\mathbb{D}_N=\{(\bm{z}^{(n)},y^{(n)})\}_{n=1}^N$ with scalar measurements $y^{(n)}\in\mathbb{R}$, the hyperparameters are typically determined by minimizing the negative log-likelihood
\begin{align}\label{eq:nll}
	-\log p(\bm{y}_N|\bm{Z}_N,\bm{\theta}) &\!=\! \frac{1}{2}\bm{y}_N^T(\bm{K}_N\!+\!\sigma_{\mathrm{on}}^2\bm{I}_N)^{-1}\bm{y}_N
	\\
	&\!+\! \frac{1}{2}\log(\det(\bm{K}_N\!+\!\sigma_{\mathrm{on}}^2\bm{I}_N)) \!+\! \frac{n}{2}\log(2\pi),\nonumber
\end{align}
where $\bm{Z}_N=\begin{bmatrix} \bm{z}^{(1)}&\cdots&\bm{z}^{(N)} \end{bmatrix}$,  $\bm{y}_N=\begin{bmatrix}y^{(1)}&\cdots&y^{(N)} \end{bmatrix}^T$ are the concatenated training samples, and $\bm{K}_N$ is the covariance matrix, whose elements are defined via $k_{N,i,j}=k(\bm{z}^{(i)},\bm{z}^{(j)}) $, $i,j=1,\ldots,N$. 
Even though the negative log-likelihood generally is a non-convex function with respect to the hyperparameters, it is optimized using gradient-based methods in practice \cite{Rasmussen2006}. 

After the hyperparameters have been determined, the posterior distribution at an arbitrary test point $\bm{z}$ can be analytically computed under the assumption of Gaussian distributed noise $\bm{\epsilon}$. This posterior is again a Gaussian distribution with mean and variance defined as
\begin{align}\label{eq:GPmean}
    \mu_{N}(\bm{z}) &= \bm{k}_N^T( \bm{z}_{N})(\bm{K}_N + \sigma_{\mathrm{on}}^{2}\bm{I}_{N})^{-1}\bm{y}_{N}\\ 
    \label{eq:GPvariance}
        \sigma_{N}^{2} (\bm{z}) &= k(\bm{z}, \bm{z})- \bm{k}_N^T(\bm{z})(\bm{K}_N \!+\! \sigma_{\mathrm{on}}^{2}\bm{I}_{N})^{-1}\bm{k}_N(\bm{z}),
\end{align}
where the kernel vector $\bm{k}_N(\bm{z})$ is defined elementwise via $k_{N,i}(\bm{z})=k(\bm{z},\bm{z}^{(i)})$, $i=1,\ldots,N$.


\subsection{Online Learning with Sliding Window Gaussian Process}\label{subsec:SW-GP}

\DecMargin{0.9em}
\begin{algorithm}[t]
    \setstretch{1.5}
    \SetInd{0.5em}{0.5em}
    \DontPrintSemicolon
	\small
	\SetKwFunction{FSOP}{update}
	\SetKwFunction{upK}{getA}
	\SetKwFunction{hypUp}{hypUpdate}
	\SetKwFunction{grad}{gradient}
	\SetKwFunction{step}{stepUpdate}
	\SetKwProg{Fn}{Function}{:}{}
	\Fn{\FSOP{$\mathbb{W}^{(N)}\!\!$, $\bm{\theta}^{(N)}\!\!$, $\bm{\Delta}^{\!(N)}\!\!$, $\bm{\nabla}_{\mathrm{nll}}^{(N)}\!\!$, $\bm{z}^{(N+1)}\!\!$, $y^{(N+1)}$} } {
	    $\mathbb{W}^{(N+1)}\!\!\gets\! \mathbb{W}^{(N)}\cup\{(\bm{z}^{(N+1)},y^{(N+1)})\}$\;
	    \If { $|\mathbb{W}^{(N+1)}| > \bar{N}$ }{
            $ \mathbb{W}^{(N+1)} \!\!\gets\! \mathbb{W}^{(N+1)}\setminus \{(\bm{z}^{(N+1-\bar{N})},y^{(N+1-\bar{N})}) \} $ \;
        }
		$\bm{\theta}^{(N+1)} \!\!\gets\! \bm{\theta}^{(N)}+\mathrm{sign}(\bm{\nabla}_{\mathrm{nll}}^{(N)})\bm{\Delta}^{(N)} $ \;
		$\bm{A}_{\mathbb{W}^{(N+1)}}\!,\bm{\alpha}_{\mathbb{W}^{(N+1)}}\!\!\gets\!\upK(\mathbb{W}^{(N+1)}\!,\bm{\theta}^{(N+1)})\!\!$~ \textcolor{lightgray}{\% \eqref{eq:AW}, \eqref{eq:alphaW}}\;
		$\bm{\nabla}_{\mathrm{nll}}^{(N+1)} \!\!\gets\! \grad(\bm{A}_{\mathbb{W}^{(N+1)}},\bm{\alpha}_{\mathbb{W}^{(N+1)}}, \bm{\theta}^{(N)}) \!$~~ \textcolor{lightgray}{\% \eqref{eq:nablaNLL}}\;
		$\bm{\Delta}^{(N+1)}\!\! \gets\!\step(\bm{\nabla}_{\mathrm{nll}}^{(N+1)}, \bm{\nabla}_{\mathrm{nll}}^{(N)}, \bm{\Delta}^{(N)} )\!\!~$~~  \textcolor{lightgray}{\% \eqref{eq:RPROP_step}}\;
	}
	\KwRet $\mathbb{W}^{(N+1)}\!\!$, $\bm{\theta}^{(N+1)}\!\!$, $\bm{\Delta}^{(N+1)}\!\!$, $\bm{\nabla}_{\mathrm{nll}}^{(N+1)}\!\!$,  $\bm{\alpha}_{\mathbb{W}^{(N+1)}}$, $\bm{A}_{\mathbb{W}^{(N+1)}}^{-1}\!\!$
	\caption{Iterative learning with SW-GPs}
	\label{alg:SW-GP}
\end{algorithm}

Even though exact Gaussian process regression as described in \cref{subsec:GP} allows closed-form expressions for the posterior mean and variance functions, it crucially suffers from a cubically growing computational complexity in the number of training samples $N$. In order to overcome this limitation, we propose sliding window Gaussian processes (SW-GP) as outlined in Alg.~\ref{alg:SW-GP}, which rely on the idea of handling continuous data streams by employing only the most recent training samples for training a GP. In particular, after observing $N\in\mathbb{N}_+$ training samples, we restrict the training data of a GP to the subset
\begin{align}\label{eq:window}
    \mathbb{W}^{(N)}=\{ (\bm{z}^{(n)},y^{(n)}) \}_{n=N+1-\bar{N}}^N,
\end{align}
where $\bar{N}\in\mathbb{N}_+$ is a constant determining the number of samples in the set $\mathbb{W}^{(N)}$. By using these windows of data $\mathbb{W}^{(N)}$ instead of the full data set $\mathbb{D}_N$, the posterior mean and variance can be computed in constant complexity using 
\begin{align}\label{eq:muW}
    \mu_{\mathbb{W}^{(N)}}(\bm{z}) &= \bm{k}_{\mathbb{W}^{(N)}}^T(\bm{z})\bm{\alpha}_{\mathbb{W}^{(N)}}\\ 
    \sigma_{\mathbb{W}^{(N)}}^{2} (\bm{z}) &= k(\bm{z}, \bm{z})- \bm{k}_{\mathbb{W}^{(N)}}^T(\bm{z})\bm{A}_{\mathbb{W}^{(N)}}^{-1}\bm{k}_{\mathbb{W}^{(N)}}(\bm{z}),
        \label{eq:sigW}
\end{align}
where 
\begin{align}\label{eq:AW}
    \bm{A}_{\mathbb{W}^{(N)}}&=\bm{K}_{\mathbb{W}^{(N)}}+\sigma_{\mathrm{on}}^2\bm{I}_{\bar{N}}\\
    \bm{\alpha}_{\mathbb{W}^{(N)}}&=\bm{A}_{\mathbb{W}^{(N)}}^{-1}\bm{y}_{\mathbb{W}^{(N)}}
    \label{eq:alphaW}
\end{align}
and we use indexes $\mathbb{W}^{(N)}$ to emphasize that $\mathbb{W}^{(N)}$ is used for computation instead of $\mathbb{D}_N$, e.g., $\bm{K}_{\mathbb{W}^{(N)}}$ is a matrix with elements $K_{\mathbb{W}^{(N)},i,j}=k(\bm{z}^{(N-\bar{N}+i)},\bm{z}^{(N-\bar{N}+j)})$, $i,j=1,\ldots\bar{N}$. Since \eqref{eq:muW} and \eqref{eq:sigW} reduce to simple scalar products and matrix-vector multiplications given $\bm{A}_{\mathbb{W}^{(N)}}^{-1}$ and $\bm{\alpha}_{\mathbb{W}^{(N)}}$, they allow a fast computation of predictions in practice. 

Therefore, we focus on the online adaptation of the hyperparameters $\bm{\theta}^{(N)}$ to new training pairs $(\bm{z}^{(N+1)},y^{(N+1)})$ in the following, which we realize by minimizing the negative log-likelihood \eqref{eq:nll} iteratively. Our numerical optimization scheme employed for this purpose is based on RPROP \cite{Riedmiller1993}, which has been demonstrated to outperform other elaborate methods while maintaining relatively low computational complexity \cite{Blum2013}. In contrast to existing work, we merely perform one update iteration for each data set $\mathbb{W}^{(N)}$ in order to limit the number of computations executed for each new data pair. This results in the adaptation rule 
\begin{align}\label{eq:hyp_Update}
    \bm{\theta}^{(N+1)} = \bm{\theta}^{(N)}+\mathrm{sign}(\bm{\nabla}_{\mathrm{nll}}^{(N)})\bm{\Delta}^{(N)},
\end{align}
where the gradient of the negative log-likelihood can be efficiently computed using 
\begin{align}\label{eq:nablaNLL}
	\nabla_{\mathrm{nll},j}^{(N)}&=\frac{\partial}{\partial \theta_{j}}(-\log p(\bm{y}_{\mathbb{W}^{(N)}}|\bm{Z}_{\mathbb{W}^{(N)}},\bm{\theta})) \nonumber\\
	&= -\frac{1}{2}\mathrm{tr}\left((\bm{\alpha_{\mathbb{W}^{(N)}}\alpha}_{\mathbb{W}^{(N)}}^{T}\!-\!\bm{A}_{\mathbb{W}^{(N)}}^{-1})\frac{\partial \bm{A}_{\mathbb{W}^{(N)}}}{\partial \theta_{j}}\right)
\end{align}
for $\frac{\partial \bm{A}_{\mathbb{W}^{(N)}}}{\partial \theta_{j}}$ determined using the partial derivatives of the kernel $k(\cdot,\cdot)$ due to \eqref{eq:AW}. 
The step size is defined recursively for each dimension $i=1,\ldots,d$ through 
\begin{align}\label{eq:RPROP_step}
    \Delta_i^{(N+1)}= 
    \begin{dcases}
        \eta^{+}\Delta_i^{(N)} ,&\quad \nabla^{N+1}_{\mathrm{nll},i} \nabla^{N}_{\mathrm{nll},i} > 0\\
        \eta^{-}\Delta_i^{(N)} ,&\quad \nabla^{N+1}_{\mathrm{nll},i} \nabla^{N}_{\mathrm{nll},i} < 0\\
        \Delta_i^{(N)} ,&\quad \nabla^{N+1}_{\mathrm{nll},i} \nabla^{N}_{\mathrm{nll},i} = 0\\
    \end{dcases}
\end{align}
for the update factors $\eta^+,\eta^-\in\mathbb{R}^+$ with $\eta^{+} \!> 1$ and $\eta^{-} \!< 1 $. The updates of the step size $\bm{\Delta}$ follow the simple idea to take larger steps when the descent direction does not change between two iterations as indicated by a positive value of $\nabla^{N+1}_{\mathrm{nll},i} \nabla^{N}_{\mathrm{nll},i}$. When this product becomes negative, the sign of the derivative has changed between iterations, such that the optimization went past a local minimum. Therefore, the step size $\bm{\Delta}$ is reduced in this case, leading to simple and intuitive updates of the hyperparameters.



\begin{remark}
We employ the gradient $\bm{\nabla}_{\mathrm{nll}}^{(N)}$ from the previous iteration for computing the hyperparameter updated \eqref{eq:hyp_Update}, which does not depend on the new data pair $(\bm{z}^{(N+1)},y^{(N+1)})$. Therefore, we can re-use the values from the previous update iteration, such that we merely have to compute $\bm{A}_{\mathbb{W}^{(N)}}$ and $\bm{\alpha}_{\mathbb{W}^{(N)}}$ once per update iteration, which significantly reduces computation time. 
\end{remark}







\subsection{Low-Pass Filtering with Gaussian Processes}\label{subsec:GP-filter}


In order to apply the presented sliding window Gaussian process approach for filtering a signal $\bm{y}^{(N)}$, we first note that the time-triggered sampling as introduced in \cref{ass:data} establishes the relationship 
\begin{align}
    N(t)=\left\lfloor\frac{t}{\tau}\right\rfloor+1,
\end{align}
such that the data window $\mathbb{W}^{(N(t))}=\{ (\bm{z}^{(n)},\bm{y}^{(n)}) \}_{n=N(t)+1-\bar{N}}^{N(t)}$ changes over time. 
Since the observations $\bm{y}^{(N(t))}$ are $d$-dimensional vectors, we define separate sets $\mathbb{W}_i^{(N(t))}$, $i=1,\ldots,d$, for each dimension of $\bm{y}^{(N(t))}$. 
For each of these sets, we compute a SW-GP mean function  $\mu_{\mathbb{W}_i^{(N(t))}}(t)$ as defined in \eqref{eq:muW}, which we concatenate in the vector $\bm{\mu}_{\mathbb{W}^{(N(t))}}(t)=\begin{bmatrix}\mu_{\mathbb{W}_1^{(N(t))}}(t)&\cdots& \mu_{\mathbb{W}_d^{(N(t))}}(t)\end{bmatrix}^T$. For notational simplicity, we use indices $\bar{N}$ to denote the values computed using $\mathbb{W}^{(N(t))}$ in the following, e.g., $\bm{\mu}_{\bar{N}}(t)=\bm{\mu}_{\mathbb{W}^{(N(t))}}(t)$. Using this notation, we obtain the filtered signal via\looseness=-1
\begin{align}
    \hat{\bm{x}}(t)=\bm{\mu}_{\bar{N}}(t).
\end{align}

\begin{remark}
When a linear relationship between the target dimensions is known, the proposed method can be directly extended to multi-output GPs using coregionalization \cite{Lederer2021c}.
\end{remark}

For showing that this filter guarantees a uniformly bounded estimation error, we first derive a closed-form expression for the regression error, which allows an intuitive separation of error sources. For this purpose, we  define the subsets of $\mathbb{W}^{(N(t))}$ containing only the first $i$ training pairs as $ \mathbb{W}_i^{(N(t))}=\{ (\bm{z}^{(n)},y^{(n)}) \}_{n=N(t)+1-\bar{N}}^{N(t)+i-\bar{N}}$. Furthermore, we use indices $i$ to denote the values computed using the training set $\mathbb{W}_i^{(N(t))}$, e.g., $\bm{\mu}_i(t)=\bm{\mu}_{\mathbb{W}_i^{(N(t))}}(t)$. 
Based on this notation, we introduce the following auxiliary result.
\begin{lemma}\label{lem:kernelinv}
The inverse of a kernel matrix 
\begin{align}\label{eq:Aw}
    \bm{A}_{\bar{N}}=\begin{bmatrix} \bm{A}_{\bar{N}-1}& \bm{k}_{\bar{N}-1}(t^{(\bar{N})})\\ \bm{k}^T_{\bar{N}-1}(t^{(\bar{N})})& k(t^{(\bar{N})},t^{(\bar{N})})+\sigma_{\mathrm{on}}^2, \end{bmatrix}
\end{align}
can be computed by
\begin{align}
    \!\bm{A}_{\bar{N}}^{-1}&\!=\!\begin{bmatrix}\bm{A}_{\bar{N}-1}^{-1}&0\\0&0\end{bmatrix}\!+\!\frac{1}{\sigma_{\bar{N}-1}^2\!(t^{(\bar{N})})\!+\!\sigma_{\mathrm{on}}^2}\!\!\begin{bmatrix}\bm{a}\\-1 \end{bmatrix}\!\!\begin{bmatrix}\bm{a}^T&-1 \end{bmatrix}\!
\end{align}
with $\bm{a}=\bm{A}_{\bar{N}-1}^{-1}\bm{k}_{\bar{N}-1}(t^{(\bar{N})})$. 
\end{lemma}
\begin{proof}
Due to the definition of $\bm{A}_{\bar{N}}$, the block inversion formula yields
\begin{align}
    &\bm{A}_{\bar{N}}^{-1}\!=\!
    \begin{bmatrix}\bm{A}_{\bar{N}-1}^{-1}+\frac{\bm{a}\bm{a}^T}{\sigma_{\bar{N}-1}^2\!\!(t^{(\bar{N})})+\sigma_{\mathrm{on}}^2}&-\frac{\bm{a}}{\sigma_{\bar{N}-1}^2(t^{(\bar{N})})+\sigma_{\mathrm{on}}^2}\\ -\frac{\bm{a}^T}{\sigma_{\bar{N}-1}^2(t^{(\bar{N})})+\sigma_{\mathrm{on}}^2}& \frac{1}{\sigma_{\bar{N}-1}^2(t^{(\bar{N})})+\sigma_{\mathrm{on}}^2}  \end{bmatrix}\!.\nonumber
\end{align}
Rearranging the terms leads to the result.
\end{proof}

Based on this result, we can derive the following theorem, which we state for scalar systems with $d=1$ for simplicity, but generalization to arbitrary dimension $d$ is straightforward.

\begin{theorem}\label{th:esterr}
Consider a scalar filtering problem for a data stream satisfying \cref{ass:data}.  Then, the estimation error between the SW-GP prediction $\hat{x}(t)=\mu_{\bar{N}}(t)$ and the true state $x(t)$ is given by
\begin{align}\label{eq:est_err}
    \hat{x}(t)-x(t)&= \Delta(t)-\nu_0^{\bar{N}} \tilde{x}^{(1)}\nonumber\\
    &+\sum\limits_{n=1}^{\bar{N}-1}\nu_n^{\bar{N}}\Delta_n+\sum\limits_{n=1}^{\bar{N}}\nu_n^{\bar{N}}\eta_n\tilde{\epsilon}^{(n)},
\end{align}
where \allowdisplaybreaks
\begin{align}
    \eta_{n} &= \frac{\sigma_{n-1}^2(\tilde{t}^{(n)})}{\sigma_{n-1}^2(\tilde{t}^{(n)})+\sigma_{\mathrm{on}}^2}\\
    \nu_n^{\bar{N}} & =  \prod\limits_{i=n}^{\bar{N}-1}-\frac{\sigma_{\mathrm{on}}^2}{\sigma_{i}^2(\tilde{t}^{(n)})+\sigma_{\mathrm{on}}^2}\\
    \Delta_n&=\tilde{x}^{(n+1)}-\tilde{x}^{(n)}+\mu_{n}(\tilde{t}^{(n)})-\mu_{n}(\tilde{t}^{(n+1)})\label{eq:Deltan}\\
    \Delta(t)&=x(t)-\tilde{x}^{(\bar{N})}+\mu_{\bar{N}}(\tilde{t}^{(\bar{N})})-\mu_{\bar{N}}(t)
\end{align}
for $\tilde{t}^{(n)}\!\!=\!t^{(N+n-\bar{N})}$, $\tilde{x}^{(n)}\!\!=\!x^{(N+n-\bar{N})}$ and $\tilde{\epsilon}^{(n)}\!\!=\!\epsilon^{(N+n-\bar{N})}$.
\end{theorem}
\begin{proof}
We prove this theorem by induction, showing that for any $j=1,\ldots,\bar{N}$ it holds that the error $e_j=\tilde{x}^{(j)}-\mu_{j}(\tilde{t}^{(j)})$ satisfies
\begin{align}\label{eq:indass}
    e_j&= \sum\limits_{n=1}^{j}\eta_n\nu_n^j\tilde{\epsilon}^{(n)}
    +\sum\limits_{n=1}^{j-1}\nu_n^j\Delta_n 
    +\nu_0^j\tilde{x}^{(1)}.
\end{align}
For the induction basis $j=1$, we have
\begin{align}
    \mu_{1}(\tilde{t}^{(1)})=\frac{\tilde{y}^{(1)}k(\tilde{t}^{(1)},\tilde{t}^{(1)})}{k(\tilde{t}^{(1)},\tilde{t}^{(1)})+\sigma_{\mathrm{on}}^2},\nonumber
\end{align}
where $\tilde{y}^{(1)}=y^{(N+1-\bar{N})}$. Since $\sigma_{0}^2(\tilde{t}^{(1)})=k(\tilde{t}^{(1)},\tilde{t}^{(1)})$, it is straightforward to see that 
\begin{align}
    e_1&=\frac{\sigma_{0}^2(\tilde{t}^{(1)})}{\sigma_{0}^2(\tilde{t}^{(1)})+\sigma_{\mathrm{on}}^2}\tilde{\epsilon}^{(1)} + \frac{\sigma_{\mathrm{on}}^2}{\sigma_{0}^2(\tilde{t}^{(1)})+\sigma_{\mathrm{on}}^2}\tilde{x}^{(1)},\nonumber
\end{align}
which proves the induction basis. For showing the induction step, note that
\begin{align}
    \bm{y}_{j}^T\bm{A}_{j}^{-1}\left(\bm{k}_{j}(\tilde{t}^{(j)})+\sigma_{\mathrm{on}}^2\bm{1}_j\right)=\tilde{y}^{(j)}.\nonumber
\end{align}
Moreover, due to \cref{lem:kernelinv}, we have
\begin{align}
    \bm{A}_{j}^{-1}\bm{1}_j=\frac{1}{\sigma_{j-1}^2(\tilde{t}^{(j)})\!+\!\sigma_{\mathrm{on}}^2}\begin{bmatrix}\bm{A}_{j-1}^{-1}\bm{k}_{j-1}(\tilde{t}^{(j)})\\-1 \end{bmatrix},\nonumber
\end{align}
such that we obtain
\begin{align}
    e_j&=
    \tilde{\epsilon}^{(j)}-\frac{\sigma_{\mathrm{on}}^2(y^{(j)}-\mu_{j-1}(\tilde{t}^{(j)}))}{\sigma_{j-1}^2(\tilde{t}^{(j)})+\sigma_{\mathrm{on}}^2}.\nonumber
\end{align}
Decomposing the right side into noise and signal components yields
\begin{align}
    e_j&=\eta_j\tilde{\epsilon}^{(j)}-\frac{\sigma_{\mathrm{on}}^2(\tilde{x}^{(j)}-\mu_{j-1}(\tilde{t}^{(j)}))}{\sigma_{j-1}^2(\tilde{t}^{(j)})+\sigma_{\mathrm{on}}^2}.\nonumber
\end{align}
Using the definition of $\Delta_n$ in \eqref{eq:Deltan}, the second term on the right-hand side results in
\begin{align}
    \!\frac{\sigma_{\mathrm{on}}^2(\tilde{x}^{(j)}\!-\!\mu_{j-1}(\tilde{t}^{(j)}))}{\sigma_{j-1}^2(\tilde{t}^{(j)})+\sigma_{\mathrm{on}}^2}&\!=\!
    \frac{\sigma_{\mathrm{on}}^2\Delta_{j-1}}{\sigma_{j-1}^2(\tilde{t}^{(j)})\!+\!\sigma_{\mathrm{on}}^2}\!+\!
    \frac{\sigma_{\mathrm{on}}^2e_{j-1}}{\sigma_{j-1}^2(\tilde{t}^{(j)})\!+\!\sigma_{\mathrm{on}}^2}.\nonumber
\end{align}
Therefore, the induction assumption \eqref{eq:indass} yields
\begin{align}
    e_j&=\eta_j\tilde{\epsilon}^{(j)}-\frac{\sigma_{\mathrm{on}}^2 \sum\limits_{n=1}^{j-1}\eta_n\nu_n^{j-1}\tilde{\epsilon}^{(n)}
     }{\sigma_{j-1}^2(\tilde{t}^{(j)})+\sigma_{\mathrm{on}}^2}\nonumber\\
     &-\frac{\sigma_{\mathrm{on}}^2\left(\Delta_{j-1}+\sum\limits_{n=1}^{j-2}\nu_n^{j-1}\Delta_n +\nu_0^{j-1}\tilde{x}^{(1)}
    \right)}{\sigma_{j-1}^2(\tilde{t}^{(j)})+\sigma_{\mathrm{on}}^2},\nonumber
\end{align}
which is identical to \eqref{eq:indass} and therefore, concludes the induction. Finally, it remains to bound the error resulting from predicting into the future, i.e., when not predicting at the training sample $\tilde{t}^{(N)}$, but at an arbitrary point $t>\tilde{t}^{(N)}$. Since the time between these two points is bounded by $\tau$ due to the time-triggered sampling scheme, it follows that
\begin{align}\label{eq:pred_err}
    x(t)-\mu_{\bar{N}}(t)=\tilde{x}^{(\bar{N})}-\mu_{\bar{N}}(\tilde{t}^{(\bar{N})})+\Delta(t)
\end{align}
Therefore, we obtain 
\begin{align}
    x(t)\!-\!\mu_{\bar{N}}(t)&\!=\!\Delta(t)\!-\!\nu_0^{\bar{N}} \tilde{x}^{(1)}\!+\!\sum\limits_{n=1}^{\bar{N}}\!\nu_n^{\bar{N}}\eta_n\tilde{\epsilon}^{(n)}\!+\!\sum\limits_{n=1}^{\bar{N}-1}\!\nu_n^{\bar{N}}\Delta_n\nonumber
\end{align}
by substituting \eqref{eq:indass} in \eqref{eq:pred_err}, which concludes the proof.
\end{proof}

This result allows a straightforward and intuitive interpretation of the estimation error \eqref{eq:est_err}, which can be essentially decoupled into two sources, namely the signal attenuation and the unfiltered noise components. Signal attenuation corresponds to the error, which would result from the SW-GP filter in the absence of any noise due to the amplitude and variation of the signal $x(\cdot)$, and it corresponds to the first three terms in \eqref{eq:est_err}. The first summand $\Delta(t)$ results from predicting the state at a future time instant. If the signal is only predicted at sampling times $t=n\tau$, $n\in\mathbb{N}$, it is trivial to see that $\Delta(n\tau)=0$. The term $\nu_0^{\bar{N}}\tilde{x}^{(1)}$ corresponds to the attenuation of the signal amplitude due to the prior distribution with mean $m(t)=0$, while $\sum_{n=1}^{\bar{N}-1}\nu_n^{\bar{N}}\Delta_n$ is the error resulting from variation in the signal. While the latter would be approximately equal to $0$ for constant signals $x(\cdot)$, the former is non-zero for non-zero signals. Both these terms crucially depend on the parameters $\nu_n^{\bar{N}}$, $n=1,\ldots,\bar{N}$, which allow intuitive insight into the estimation process. It can be straightforwardly observed that $\nu_n^{\bar{N}}\approx 1$ if $\sigma_{\mathrm{on}}\gg \sigma_{n}(\tilde{t}^{(n+1)})$ for all $i$. This implies that the signal and mean variations $\Delta_n$ are almost summed up, which can yield a high estimation error. However, this behavior is to be expected since large observation noise standard deviations $\sigma_{\mathrm{on}}$ mean that the signal is heavily smoothed, thereby, strongly attenuating high frequency signals. In contrast, small values $\sigma_{\mathrm{on}}\approx \sigma_{n}(\tilde{t}^{(n+1)})$ are beneficial for signal reconstruction since they result in $|\nu_n^{\bar{N}}|\appropto \frac{1}{2^{N-n}}$. Therefore, variation in the signal $x(\cdot)$ several time steps before the estimation barely has any effect on the reconstruction error. 
While a small value of $\sigma_{\mathrm{on}}$ is beneficial for low signal attenuation, it also barely attenuates the noise. This becomes clear in the last term in \eqref{eq:est_err}, which represents the unfiltered noise components. For $\sigma_{\mathrm{on}}\approx \sigma_{n}(\tilde{t}^{(n+1)})$, it follows that $\eta_n\approx 1$, such that noise is almost directly passed through to the estimated state $\hat{x}(t)$. In contrast, high observation noise standard deviations $\sigma_{\mathrm{on}}\gg \sigma_{n}(\tilde{t}^{(n+1)})$ lead to $\eta_n\approx 0$, such that noise is almost completely removed from the signal.  

This clearly demonstrates the conflict between noise attenuation and signal recovery in SW-GP filters, which is a well-known problem in systems theory and filter design~\cite{Orfanidis2016}. This issue is commonly resolved by defining frequency ranges, where either of those two behaviors is dominant, leading to, e.g., low-pass filters. In SW-GP filters, we pursue a similar approach by exploiting the hyperparameters $\bm{\theta}$, which allow us to directly modify the observation noise standard deviation $\sigma_{\mathrm{on}}$, but also $\sigma_{n}(\tilde{t}^{(n)})$ through, e.g., the length scales $\bm{l}$ when using a squared exponential kernel \eqref{eq:SE}. Thereby, the behavior can be tuned for noise attenuation or low signal disturbance, depending on the measurements $\bm{y}^{(n)}$. However, in contrast to traditional filtering techniques, this approach has the advantage that hyperparameters can be automatically tuned using log-likelihood maximization. In fact, it is even possible to adapt the hyperparameters online, such that adaptive filters can be straightforwardly implemented using Alg.~\ref{alg:SW-GP}.

While \cref{th:esterr} allows a straightforward interpretation, it cannot be calculated in practice as the error realizations~$\tilde{\epsilon}^{(n)}$, $n\!=\!1,\ldots,\bar{N}$ are generally unknown. Through additional assumptions on the observation noise and the employed kernel, these issues can easily be overcome. For example, assuming bounded noise results in the following error bound.
\begin{corollary}
Consider a scalar filtering problem for a data stream satisfying \cref{ass:data}. Assume a Lipschitz continuous kernel is employed for computing the sliding window Gaussian process mean \eqref{eq:muW} and assume that the noise is bounded by a constant $\bar{\epsilon}\in\mathbb{R}_{0,+}$, i.e., $|\epsilon^{(i)}|\leq \bar{\epsilon}$ for all $i\in\mathbb{N}$. Then, the estimation error between $\hat{x}(t)=\mu_{\bar{N}}(t)$ and $x(t)$ is uniformly bounded with upper bound
\begin{align}
    c(t)=|\nu_0^{\bar{N}} \tilde{x}^{(1)}|
    +\sum\limits_{n=1}^{\bar{N}}|\nu_n^{\bar{N}}|(\eta_n\bar{\epsilon}+(L_{\mu_{n}}\!+L_x)\tau),
\end{align}
where $L_{\mu_{n}}$ and $L_x$ denote the Lipschitz constants of the SW-GP mean functions and the signal, respectively.
\end{corollary}
\begin{proof}
Lipschitz continuity of kernel used in the sliding window GP implies Lipschitz continuity of the mean $\mu_{n}(\cdot)$ for all $n=1,\ldots,\bar{N}$ \cite{Lederer2019} and boundedness of $f(\cdot,\cdot,\cdot)$ implies Lipschitz continuity of $x(\cdot)$. Therefore, we obtain 
\begin{align}
    \Delta_n&\leq (L_{\mu_{n}}\!+L_x)\tau)\\
    \Delta(t)&\leq (L_{\mu_{n}}\!+L_x)\tau).
\end{align}
The remainder of the proof follows from \cref{th:esterr}.
\end{proof}

\begin{remark}\label{remark:linear}
If the estimate $\hat{\bm{x}}(t)=\mu_{\bar{N}}(t)$ is determined only at sampling times $t=N\tau$, we can write it as a linear combination of previous measurements $\bm{y}^{(n)}$, $n=N-\bar{N}+1,\ldots,N$, e.g., for $d=1$ we have $\hat{x}(N\tau)=\bm{c}^T\bm{y}_{\bar{N}}$, 
where $\bm{c}^T=\bm{k}_{\bar{N}}(N\tau)(\bm{K}_{\bar{N}}+\sigma_{\mathrm{on}}^2\bm{I}_{\bar{N}})^{-1}$.  Therefore, the SW-GP filter is a linear finite impulse response filter, whose filter coefficients are adapted to the input signal. Note that for stationary kernels $k(\cdot,\cdot)$ \cite{Rasmussen2006} and fixed hyperparameters, the filter coefficients are even constant, such that moving average and Sawitzky-Golay type filters \cite{Orfanidis2016} can be realized.
\end{remark}

\section{Numerical Evaluation}\label{sec:eval}

In order to demonstrate the low-pass behavior and adaptation capabilities of SW-GPs, we execute three different numerical experiments. In \cref{subsec:freqResp}, we illustrate the dependency of the SW-GP filters low-pass behavior using its frequency response. In \cref{subsec:noiseAtt}, we demonstrate the adaptivity of the SW-GP due to the online hyperparameter optimization, which allows noise attenuation over far wider frequency ranges than existing filters. Finally, we employ the SW-GP filter in a robot control simulation to illustrate its practical applicability.\looseness=-1

\subsection{Frequency Response Analysis}\label{subsec:freqResp}

For demonstrating the low-pass behavior of SW-GP filters and its dependency on the hyperparameters, we use a squared exponential kernels \eqref{eq:SE} with fixed hyperparameters $\sigma_f=1$, $\sigma_{\mathrm{on}}=\sqrt{0.1}$ and use a window size $\bar{N}=200$. Moreover, we exemplarily choose different length scales $l=0.1$, $l=0.4$ and $l=0.01$ to illustrate their impact on the cutoff frequency. We apply it to sinusoidal input signals $x(t)=\sin(2\pi ft)$ with frequencies in the range of $10^{-2}$ \si{Hz} to $10^{2}$ \si{Hz}. The signals are sampled with a rate of $1kHz$, and the SW-GP filter is updated and evaluated at the same rate. As exemplarily illustrated in \cref{fig:filtSignal} for an input signal with $f=10^{1.5}\approx 31.62$ \si{Hz}, the filtered signal is clearly sinusoidal again which is a direct consequence of \cref{remark:linear}. 

\begin{figure}
    \centering
    \def\file{plots/filteredSignal.txt}
	\vspace{0.2cm}
	\begin{tikzpicture}
	\begin{axis}[xlabel={time [s]}, ylabel={signal},
	xmin=0, ymin = -1.2, xmax = 0.2, ymax = 1.2, height =4cm,
	x label style={yshift=0.2cm}, legend pos=north east, name=plot1,legend columns=2]
	\addplot[black, thick]	table[x = tfull, y  = Xfull]{\file};
	\addplot[blue,dashed, thick]	table[x = tfull, y  = Xfilt]{\file};
	
	\legend{original signal, filtered signal};
	\end{axis}
	\end{tikzpicture}
	\vspace{-0.3cm}
    \caption{When filtering a sinusoidal signal with the SW-GP, the resulting output signal is sinusoidal with reduced amplitude and small phase shift.}
    \label{fig:filtSignal}
\end{figure}
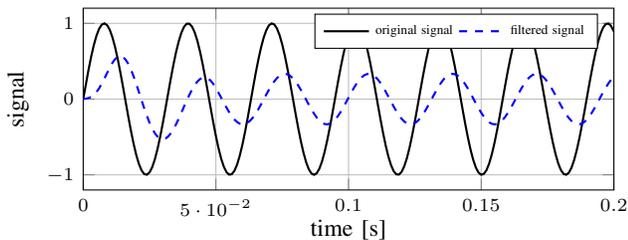

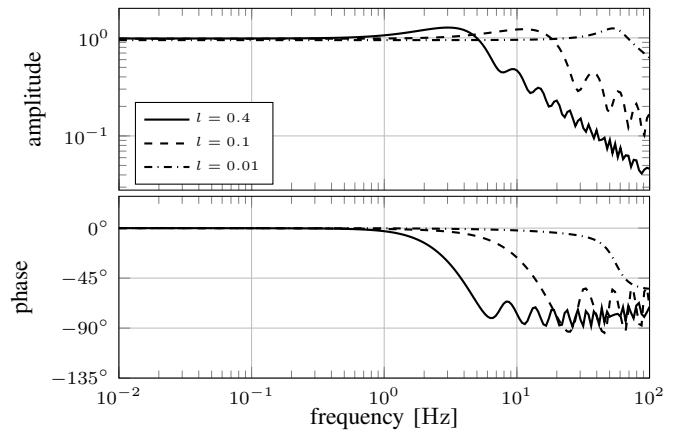
\begin{figure}
    \centering
    \def\file{plots/nominal_low_pass.txt}
	\vspace{0.2cm}
	\begin{tikzpicture}
	\begin{loglogaxis}[ylabel={amplitude},
	xmin=0.01, ymin = 0, xmax = 100, ymax = 2, height =4cm,
	x label style={yshift=0.2cm}, legend pos=south west, name=plot1,xticklabels={,,}]
	\addplot[black,thick]	table[x = f, y  = amp04_GP]{\file};
	\addplot[black,dashed,thick]	table[x = f, y  = amp01_GP]{\file};
	\addplot[black,dashdotted,thick]	table[x = f, y  = amp001_GP]{\file};
	\legend{$l=0.4$, $l=0.1$, $l=0.01$};
	\end{loglogaxis}
	\begin{semilogxaxis}
	[xlabel={frequency [\si{Hz}]},ylabel={phase},
	xmin=0.01, ymin = -2.3562, xmax = 100, ymax = 0.5, height =4cm,
	x label style={yshift=0.2cm}, legend pos=south west, at=(plot1.south), anchor=south,
	yshift=-2.5cm,
	ytick={0,-0.7854,-1.5708,-2.3562},yticklabels={$0^{\circ}$,$-45^{\circ}$,$-90^{\circ}$,$-135^{\circ}$}]
	\addplot[black,thick]	table[x = f, y  = phase04_GP]{\file};
	\addplot[black,dashed,thick]	table[x = f, y  = phase01_GP]{\file};
	\addplot[black,dashdotted,thick]	table[x = f, y  = phase001_GP]{\file};
	\end{semilogxaxis}
	\end{tikzpicture}
	\vspace{-0.7cm}
    \caption{The SW-GP filter shows a low pass behavior for sinusoidal input signals with increasing cutoff frequency for smaller length scales $l$.}
    \label{fig:lowPassBehavior}
\end{figure}

Therefore, we analyze the frequency response of the SW-GP filter using a Bode plot, which is depicted in \cref{fig:lowPassBehavior}. It can be clearly seen that the SW-GP exhibits a low-pass behavior, with the signal amplified near its cut-off frequency and non-continuous decay. Furthermore, it can be observed that the cutoff frequency strongly depends on the length scale, which leads to, e.g., a cutoff at approximately \SI{20}{Hz} for $l=0.1$. In fact, we empirically found that the cut-off frequency $f_c$ of the SW-GP filter can be approximately computed from its hyperparameters via $f_c\approx \frac{2}{l}$ for large enough windows $\bar{N}$, where $l$ is the length scale of the squared exponential kernel \eqref{eq:SE}. 
Since the proposed hyperparameter optimization scheme in \cref{subsec:SW-GP} can tune the hyperparameters online, this relationship allows to automatically adapt the cut-off frequency to the observed signal. This beneficial feature is investigated in detail in the following subsection.

\subsection{Adaptive Noise Attenuation}\label{subsec:noiseAtt}

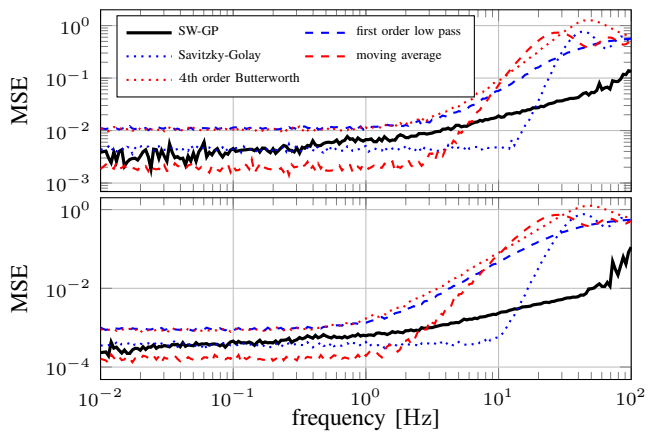
\begin{figure}[t]
    \centering
	\vspace{0.2cm}
	\begin{tikzpicture}
	\def\file{plots/Gaussian_noise_small.txt}
	\begin{loglogaxis}[ ylabel={MSE},
	xmin=0.01, ymin = 0, xmax = 100, ymax = 2, height =4cm,
	x label style={yshift=0.2cm}, legend pos=north west, name=plot1,legend columns=2, xticklabels={,,}, name=plot1]
	\addplot[black,very thick]	table[x = f, y  = mse_DWGP]{\file};
	\addplot[blue,dashed, thick]	table[x = f, y  = mse_LP]{\file};
	\addplot[blue, dotted, thick]	table[x = f, y  = mse_SG]{\file};
	\addplot[red, dashed, thick]	table[x = f, y  = mse_MM]{\file};
	\addplot[red, dotted, thick]	table[x = f, y  = mse_Butter]{\file};
	
	\legend{SW-GP, first order low pass,
	Savitzky-Golay, moving average, $4$th order Butterworth};
	\end{loglogaxis}
	\def\file{plots/Uniform_noise_small.txt}
	\begin{loglogaxis}[xlabel={frequency [\si{Hz}]}, ylabel={MSE},
	xmin=0.01, ymin = 0, xmax = 100, ymax = 2, height =4cm,
	x label style={yshift=0.2cm}, legend pos=north west, legend columns=2,at=(plot1.south), anchor=south,
	yshift=-2.5cm,]
	\addplot[black,very thick]	table[x = f, y  = mse_DWGP]{\file};
	\addplot[blue,dashed, thick]	table[x = f, y  = mse_LP]{\file};
	\addplot[blue, dotted, thick]	table[x = f, y  = mse_SG]{\file};
	\addplot[red, dashed, thick]	table[x = f, y  = mse_MM]{\file};
	\addplot[red, dotted, thick]	table[x = f, y  = mse_Butter]{\file};
	
	\end{loglogaxis}
	\end{tikzpicture}
	\vspace{-0.7cm}
    \caption{The estimation error for SW-GPs grows continuously with the frequency for signals perturbed by Gaussian noise (top) and uniform noise (bottom). In contrast, classical filter approaches exhibit a low error for low frequencies, but after the cut-off frequency, no signal can be recovered.}
    \label{fig:Gausianlow}
\end{figure}

In order to evaluate the adaptive attenuation of noise of the SW-GP filter over a range of input signal frequencies, we employ the sinusoidal input signals from \cref{subsec:freqResp}, but perturb them by noise. We run the SW-GP filter for each frequency for \SI{1}{s} and average over $20$ repetitions of the experiment. The hyperparameters of the SW-GP are initialized as $\sigma_f=1$, $l=0.1$ and $\sigma_n=\sqrt{0.1}$ and adapted using $\eta^+=1.2$ and $\eta^-=0.5$ in all simulations. The performance of the SW-GP filter is evaluated in terms of the mean squared error (MSE) and compared to a first order low-pass, a fourth order Butterworth, a moving average, and a third order Sawitzki-Golay filter. All finite impules response filters consider a window of data with $\bar{N}=200$.

The resulting MSE curves for Gaussian noise with variance $\sigma_n^2=0.1$ and uniform noise with bounds $\pm 0.5\sqrt{0.1}$ are depicted in \cref{fig:Gausianlow}. It can be observed, that the SW-GP filter outperforms the infinite impulse response filters over the whole frequency range. While it yields slightly higher errors than the finite impulse response filter approaches for low frequencies, the SW-GP filter is capable of adapting to the signal frequency, such that it attenuates the signal significantly less in higher frequency domains. Overall, this leads to a MSE curve which slowly grows, while classical filters are limited to a prescribed frequency band for the input signal. Therefore, SW-GP filters can be used in applications, where prior knowledge about the frequency distribution of the input signal is not available at the time of the filter design.

\subsection{Low-Pass Filtering for Robot Control}\label{subsec:robotControl}

We demonstrate the practical applicability of the proposed SW-GP filter in a simulation of a planar robotic manipulator with two degrees of freedom \cite{Murray1994}, which is controlled using a PD control law $\bm{u}=-k_p(\bm{q}-\bm{q}_{\mathrm{ref}})-k_d(\dot{\bm{q}}-\dot{\bm{q}}_{\mathrm{ref}})$ with state $\bm{x}=[\bm{q}^T\ \dot{\bm{q}}^T]^T$ and reference $\bm{x}_{\mathrm{ref}}=[\bm{q}^T_{\mathrm{ref}}\ \dot{\bm{q}}^T_{\mathrm{ref}}]^T$. The gains are set to $k_p=100$, $k_d=10$ and the control law is run at a frequency of \SI{1}{kHz} using zero order hold. The reference signal is chosen to be sinusoidal with increasing frequency, which is realized through the function $\bm{q}_{\mathrm{ref}}=[\sin(0.1t^2)\ 0.5\sin(0.1t^2)]^T$. We assume access to measurements of the joint angles $\bm{q}$, which are perturbed by Gaussian noise with noise standard deviation $\sigma_{\mathrm{on}}=0.1$. Angular velocities are obtained through numerical differentiation with finite differences. The SW-GP is initialized as in \cref{subsec:noiseAtt}, but we consider only windows with $\bar{N}=50$. In our simulations, this results in average update and prediction times of \SI{1.7}{ms} and \SI{0.09}{ms} with standard deviations of \SI{267}{\mu s} and \SI{11}{\mu s}, respectively. Therefore, updates of the SW-GP are performed at \SI{333}{Hz}, while predictions are computed at \SI{1}{kHz}, i.e., the SW-GP is updated every third measurement.\looseness=-1

The resulting mean squared error between the state $\bm{x}$ of the system controlled by noiseless and filtered signal averaged over $20$ simulation runs is illustrated in \cref{fig:robotSim}. In comparison with a controller employing  the unfiltered signal, the improvement due to the SW-GP filter can clearly be seen. While the error using the filtered signal slightly increases at the beginning, this is in accordance with the results from \cref{subsec:noiseAtt} as the reference trajectory $\bm{q}_{\mathrm{ref}}$ exhibits a smaller frequency for small $t$. However, the overall increase of the MSE over time is minor. This clearly demonstrates the applicability of SW-GP filters in more complex systems such as control architectures, where the filtered signal has a direct effect on the evolution of the signal itself.

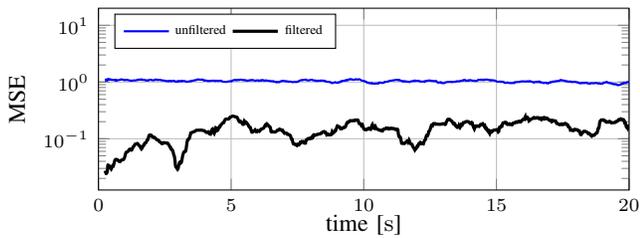
\begin{figure}[t]
    \centering
    \def\file{plots/robotSim.txt}
	\vspace{0.2cm}
	\begin{tikzpicture}
	\begin{semilogyaxis}[xlabel={time [\si{s}]}, ylabel={MSE},
	xmin=0, ymin = 0, xmax = 20, ymax = 20, height =4cm,
	x label style={yshift=0.2cm}, legend pos=north west, name=plot1,legend columns=2]
	\addplot[blue, thick]	table[x = t, y  = mse_unfilt]{\file};
	\addplot[black,very thick]	table[x = t, y  = mse_filt]{\file};
	
	\legend{unfiltered, filtered};
	\end{semilogyaxis}
	\end{tikzpicture}
	\vspace{-0.7cm}
    \caption{The SW-GP strongly attenuates the measurement noise, such that the control error induced by noisy measurements is significantly reduced.}
    \label{fig:robotSim}
\end{figure}

\section{Conclusion}\label{sec:conc}

In this paper, we propose Sliding Window Gaussian process filters, which allow adaptive low-pass filtering of signals. By training a Gaussian process on a subset of previous measurements, updates and predictions can be computed with low computational complexity, such that Gaussian processes become applicable in filtering applications. Moreover, an efficient online hyperparameter optimization scheme leads to an automatic adaptation to the input signal, such that no manual parameter tuning is necessary. Finally, we prove that the proposed method guarantees a bounded estimation error using simple linear algebra identities, and demonstrate its flexibility and efficiency in several simulations.


\bibliographystyle{IEEEtran}
\bibliography{myBib}

\end{document}